\documentclass[10pt]{article} %%Margins must be 1 inch, top & bottom 1 in%%
\usepackage{graphicx}
\usepackage{amsmath, amsthm}
\usepackage{amsmath}
\usepackage{amsfonts}
\usepackage{amssymb}
\usepackage{mathrsfs}
\usepackage{verbatim}
\usepackage{wrapfig}
\usepackage{pgf}
\usepackage{subfigure}
\usepackage{color}
\usepackage{setspace}

\newtheorem{thm}{Theorem}[section]
\newtheorem{cor}[thm]{Corollary}
\newtheorem{lemma}[thm]{Lemma}
\newtheorem{prop}[thm]{Proposition}
\newtheorem{claim}[thm]{Claim}

\theoremstyle{definition}
\newtheorem*{remm}{Remark}
\newtheorem{remark}[thm] {Remark}

\newtheorem{defn}[thm]{Definition}
\newtheorem{ex}[thm]{Example}

\newcommand{\e}{\varepsilon}

\begin{document} 
\title{\Large{On Convergence to SLE$_6$ II: \\\large{Discrete Approximations and Extraction of Cardy's Formula for General Domains}}}
\date{}
\author
{\normalsize{I.~Binder$^1$, L.~Chayes$^2$, H.~K.~Lei$^2$}
\thanks{\copyright\, 2010 by I. Binder, L.~Chayes and H.~K.~Lei.  Reproduction, by any means, of the entire article for non-commercial purposes is permitted without charge.}}
\maketitle

\vspace{-4mm}
\centerline{${}^1$\textit{Department of Mathematics, University of Toronto}}
\centerline{${}^2$\textit{Department of Mathematics, UCLA}}
 
\begin{quote}
{\footnotesize {\bf Abstract: }}Following the approach outlined in \cite{stas}, convergence to SLE$_6$ of the Exploration Processes for the correlated bond--triangular type models studied in \cite{cardy} is established in \cite{pt1} and the present work.   In this second part, we focus on establishing Cardy's Formula for general domains.

%Moreover, the proof of convergence applies in the context of general critical 2D percolation models and for general domains, under the stipulation that Cardy's Formula can be established for domains in this generality.  

\end{quote}

\section{Introduction}

In this note we wish to establish the validity of Cardy's Formula for crossing probabilities in a general (finite) 
domain $\Omega \subset \mathbb C$, clarifying certain notions concerning discretization and extraction of appropriate boundary values.  While these issues have been addressed to various extents in e.g., \cite{stas_perc}, \cite{werner_perc}, \cite{cn}, \cite{br}, \cite{rath}, and may seem quite self--evident -- at least for nice (i.e.,  Jordan) domains, a complete and unified treatment for general domains 
appears to be absent.  Moreover, aside from $\ae$sthetic appeal, the generality that appears here is certainly needed for the approach of proving convergence to SLE$_6$ outlined in \cite{stas} (see also \cite{werner_perc}) and carried out in \cite{pt1}.  Our efforts will culminate in the establishment of Theorem \ref{cardy_sup} and Corollary \ref{point_equi_cont}  (which is stated in \cite{pt1} as Lemma 2.6). 

Since it is our intention that this note be self--contained, let us first review the methodology -- introduced in \cite{stas_perc} and adapted to the models in \cite{cardy} (see also \cite{pt1}, $\S$4.1) -- by which Cardy's Formula can be extracted.
% for ease of exposition we will discuss the original hexagonal tiling problem.  
%
At the level of the continuum we are interested in a domain $\Omega \subset \mathbb C$ which is a conformal triangle with boundary components $\{\mathcal A, \mathcal B, \mathcal C\}$ and marked prime ends (boundary points) $\{a, b, c\}$ -- all in counterclockwise order -- which represent the intersection of neighboring components.  At the level of the lattice, at spacing $\varepsilon$, we consider an approximate domain $\Omega_\varepsilon$, in which the percolation process occurs and which tends -- in some sense -- to $\Omega$ as $\varepsilon \rightarrow 0$.  At the $\varepsilon$--scale, the competing (dual) percolative forces will be denoted, as is traditional, by ``yellow'' and ``blue''.

Let $z$ be an interior point (e.g., a vertex) in $\Omega_\e$.  We define the discrete crossing probability function $u_\varepsilon^B(z)$ to be probability that there is a blue path connecting $\mathcal A$ and $\mathcal B$, separating $z$ from $\mathcal C$, with similar definitions for $v_\varepsilon^B(z)$ and $w_\varepsilon^B(z)$ along with yellow versions of these functions.  For these objects, standard arguments show that subsequential limits exist; two seminal ingredients are required:  First, they converge to harmonic functions with a particular conjugacy relation between them in the interior and second they satisfy certain (``obvious'')
boundary values.  With these ingredients in hand it can be shown  that the limiting functions are the so called Carleson--Cardy functions.  E.g.,
\[ \lim_{\varepsilon \rightarrow 0} u_\varepsilon^Y = u,\]
and similarly for the $v$'s and $w$'s where, e.g., according to \cite{beffara}, the functions $u, v, w$ are such that 
$$F: = u + e^{2\pi i/3} v + e^{-2\pi i/3} w$$ is the unique conformal map from $\Omega$ to the equilateral triangle formed by the vertices 1, $e^{\pm2\pi i/3}$.  This is equivalent to Cardy's formula.  

We enable this program for a general class of domains \emph{and} their discrete approximations which is suitable for our uses in \cite{pt1}, Lemma 2.6/Corollary \ref{point_equi_cont}.

\begin{remm}
The appropriate discrete conjugacy relations for the $u_\e, v_\e$ and  $w_\e$ have only been established for the models in \cite{stas_perc} and \cite{cardy}.  However, since the RSW estimates are purportedly universal and actually hold for any reasonable critical 2D percolation model, in principle we always have limiting functions $u, v, w$ with some boundary values.  Hence most of the content of the present work should apply. 
However,  certain provisos and clarifications will be required; see Remark \ref{bv_remark}.
\end{remm}

%For ease of exposition we will continue to focus exclusively on tiling/site percolation problems.  Other percolation models, e.g., bond percolation models, can be treated by analogous methods.  For a tiling percolation problem, tiles are considered to be neighbors if they have a portion of an edge in common; the boundary of a connected cluster of tiles is a polygonal path of \emph{medial} edges.  Further, the arguments do not change if we assume we have more boundary pieces $\{\mathcal A, \mathcal B, \mathcal C, \dots\}$.  
%
In the ensuing arguments we will have many occasion to make use of the uniformization map $\varphi: \mathbb D \rightarrow \Omega$ (where $\mathbb D$ denotes the unit disk) normalized so that say $\varphi(0) = z_0 \in \Omega$ for some point $z_0$ well in the interior of $\Omega$ and $\varphi^\prime(0) > 0$.  We will also identify points on $\partial \mathbb D$ with boundary prime ends of $\partial \Omega$, via the Prime End Theorem.  We refer the reader to e.g., \cite{pommerenke} for such issues.  Finally, the reader may wish to keep in mind that the reason for addressing most of the issues herein is for application to the case where the curves/slits under consideration are percolation interfaces/explorer paths; for discussions on this topic we refer the reader to \cite{pt1}.

\section{The Carath\'eodory Minimum}
We start by dissecting the well--known Carath\'edory convergence, mainly to phrase it in terms of more elementary conditions more suitable for our purposes.  The reader can find similar conditions/discussions in e.g., Section 1.4 of \cite{pommerenke}.

Our general situation concerns a sequence of domains $(\Omega_n)$ which converge in some sense to the limiting $\Omega$ along with functions ($u_n, v_n,w_n)$ converging to a harmonic triple $(u,v,w)$ satisfying the appropriate conjugacy relations.  As a minimal starting point let us consider the following pointwise (geo)metric conditions for domain convergence:

%We have a situation where we have domains $\Omega_n$ converges to $\Omega$ in some way and we have functions $u_n$ which should converge to the continuum function $u$.  

\begin{itemize}
\item[($i_I$)] If $z \in \Omega$, then $z \in \Omega_n$ for all $n$ sufficiently large.
\item[($i_{II}$)] If $z_n \in \Omega_n^c$, then all subsequential limits of $(z_n)$ must lie in $\Omega^c$.
\item[($e$)] For all $z \in \Omega^c$ (including, especially $\partial \Omega$) there exists some sequence $z_{n_k} \in \Omega_{n_k}^c$ such that $z_{n_k} \rightarrow z$. 
\end{itemize}

Conditions ($i_I$) and ($i_{II}$) ensure that limiting values of $u$, $v$ and $w$ in (the interior of) $\Omega$ can be retrieved and are defined by values of $u_n$ inside $\Omega_n$ whereas condition ($e$) implies that $\Omega_n$'s don't converge to a domain strictly larger than $\Omega$, so that the boundary values of $u$ on $\partial \Omega$ might actually correspond to (the limit of) boundary values of $u_n$ on $\Omega_n$.  Indeed, these preliminary conditions turn out to be equivalent to Carath\'eodory convergence (see e.g., \cite{duren}; although in our context we will actually not have occasion to use convergence of the relevant uniformization maps).  More precisely, first we have the following result, whose proof is elementary (and we include for completeness):
  
\begin{prop}\label{metric_cond}
Consider domains $\Omega_n, \Omega \subset \mathbb C$ all containing some point $z_0$.  Then the following are equivalent:
\begin{enumerate}
\item If $K$ is compact and $K \subset 
\Omega$, then $K \subset \Omega_n$ for all but finitely many $\Omega_n$. 
\item \emph{($i_I$)} For all $z \in \Omega$, $z \in \Omega_n$ for all but finitely many $\Omega_n$.

 \emph{($i_{II}$)} If $z_n \in \Omega_n^c$, then all subsequential limits of $(z_n)$ must lie in $\Omega^c$.
 \item If $z \in \Omega$, and $\delta < d(z, \partial \Omega)$, then $B_\delta(z) \subset \Omega_n$, for all but finitely many $\Omega_n$.
\end{enumerate}
\end{prop}
\begin{proof}
$1 \Rightarrow 2$) To see ($i_I$) suppose $z \in \Omega$ and $d(z, \partial \Omega) > \delta$, then $\overline{B_\delta(z)} \subset \Omega$ and is compact and hence we have $\overline{B_\delta(z)} \subset \Omega_n$ for all $n$ sufficiently large and hence $z \in \Omega_n$ for all $n$ sufficiently large; conversely,   To see ($i_{II}$), suppose $z_n \rightarrow z$ with $z_n \in \Omega_n^c$ and suppose towards a contradiction that $z \in \Omega$.  Then again arguing as before, $\overline{B_\delta(z)} \subset \Omega_n$ for $n$ sufficiently large, but then $z_n \in B_\delta(z)$ also for $n$ even larger, which implies that these $z_n \in \Omega_n$, a contradiction.

$2 \Rightarrow 3$) Again suppose $d(z, \partial \Omega) > \delta$ so that $\overline{B_\delta(z)} \subset \Omega$.  If it is not the case that $B_\delta(z) \subset \Omega_n$ for $n$ sufficiently large, then we can find a sequence $z_n \in B_\delta(z) \cap \Omega_n^c$.  Since $\overline{B_\delta(z)}$ is compact, there exists a subsequential limit point $z_{n_k} \rightarrow z_*$, but then by ($i_{II}$), $z_* \notin \Omega$, contradicting $\overline{B_\delta(z)} \subset \Omega$.

$3 \Rightarrow 1$)  Let $K \subset \Omega$ be compact.  We can cover $K$ by $K \subset \bigcup_{x \in K} B_{\delta_x}(x)$, with $\delta_x < d(x, \partial \Omega)$. 
By the assumed compactness, there is a finite subcover $K \subset \bigcup_{i=1}^k B_{\delta_{x_i}}(x_i)$.  By 3), for $1 \leq i \leq k$, there exists $N_i$ such that $B_{\delta_{x_i}}(x_i) \subset \Omega_n$ for all $n \geq N_i$, and hence it is the case that $K \subset \Omega_m$ for all $m > \max\{N_1, N_2, \dots, N_k\}$.
\end{proof}

Now the notion of \emph{kernel convergence} -- which in our setting of bounded, simply connected domains is, by the theorem of Carath\'edory equivalent to Carath\'edory convergence, i.e., convergence uniformly on compact sets of the corresponding uniformization maps (see e.g., \cite{duren}, Theorem 3.1) -- requires, in addition (specifically to condition 1 in the above Proposition), that $\Omega$ is the \emph{largest} (simply connected) domain satisfying the above conditions.  The addition of condition ($e$) indeed correspond to maximality; arguments similar to those just presented easily lead to the following (whose proof is elementary and is also included for completeness): 

\begin{prop}
The conditions ($i_I$), ($i_{II}$), ($e$) are equivalent to $\Omega_n$ converging to $\Omega$ in the sense of kernel convergence.  
\end{prop}
\begin{proof}
In light of the above discussion, it is sufficient to show that the condition ($e$) is equivalent to the maximality condition on $\Omega$ required by kernel convergence.  

$\Rightarrow$) Suppose $\Omega$ is not maximal and hence $\Omega \subsetneq \Omega^\prime$ where $\Omega^\prime$ satisfies ($i_I$) and ($i_{II}$).  It must be the case then there is a point $z \in \partial \Omega \cap \Omega^\prime$.  By condition ($e$) there exists $z_{n_k} \rightarrow z$ with $z_{n_k} \in \Omega_{n_k}^c$, but condition ($i_{II}$) for $\Omega^\prime$ implies that $z \in (\Omega^\prime)^c$, a contradiction. 

$\Leftarrow$)  Conversely, suppose $\Omega$ is maximal and assume towards a contradiction that $\Omega$ does not satisfy ($e$), so that there exists some point $z \in \Omega^c$ and some $\delta > 0$ such that $B_\delta(z) \subset \Omega_n$ for all $n$ sufficiently large.  By the maximality of $\Omega$, it must be the case that $\overline {B_\eta(z)} \subset \Omega$ for any $\eta < \delta$, which implies in particular that $z \in \Omega$, a contradiction.   
\end{proof}

As is perhaps already clear, Carath\'edory convergence is insufficient for our purposes: Since the functions $u, v, w$ must acquire prescribed boundary values on separate pieces of $\partial \Omega$, it is manifest that (some notion of) separate convergence of the corresponding pieces of the boundary in $\partial \Omega_n$ will be required.
Special attention is needed for the cases of domains with slits -- which are of seminal importance when we consider the problem of convergence to SLE$_\kappa$.  The situation is in fact rather subtle: Note that in both Figure \ref{bad1} and Figure \ref{swallow}, we have that $\Omega_\e$ Carath\'edory converges to $\Omega$, but whereas the situation in Figure \ref{bad1} disrupts establishment of the proper boundary value, the situation in Figure \ref{swallow} is perfectly acceptable (see Remarks \ref{int_remark} and \ref{sup_remark}).

\section{Interior Approximations}
%As will become clear, Carath\'edory convergence is insufficient for our purposes: Since the functions $u, v, w$ must acquire prescribed boundary values on separate pieces of $\partial \Omega$, it is manifest that (some notion of) separate convergence of the corresponding pieces of the boundary in $\partial \Omega_n$ will be required.
%Special attention is needed for the cases of domains with slits -- which are of seminal importance when we consider the problem of convergence to SLE$_\kappa$.  

We will begin by considering the \emph{interior} approximations, where $\Omega_\e \subset \Omega$ for all $\e$.  For earlier considerations along these lines, see \cite{duffin} and \cite{ferrand}.  
Here, the crucial advantage is that all domains can be viewed under a \emph{single} uniformization map; this allows for relatively simple resolution of all concerns of a geometric/topological nature.  Moreover,  this appears to be the simplest setting for the purposes of establishing Cardy's Formula in a fixed (static) domain (see especially Example \ref{can_app} below).  In particular, for circumstances where this is all that is of interest, the reader is invited to skip the next section altogether.  We start with

\begin{defn}[Interior Approximations]\label{int_app}
We call $(\Omega_\e^\bullet)$ an \emph{interior approximation} to $\Omega$ if:

(I) The domains $\Omega_\e^\bullet$ consist of one or more (graph) connected components, each of which is bounded by a closed polygonal path, and the union of all such polygonal paths we identify as the boundary $\partial \Omega_\e^\bullet$.  In particular, $\partial \Omega_\e^\bullet$ consists exclusively of polygonal edges each of which is a portion of the border of an element in $(\Omega_\e^\bullet)^c$.

\vspace{2mm}
(II) The boundary $\partial \Omega_\e^\bullet$ is divided disjoint segments, denoted by $\mathcal A_\e, \mathcal B_\e, \dots$ in (rough) correspondence with the (finitely many) boundary components $\mathcal A, \mathcal B, \dots$ of the actual domain $\Omega$.  In case $\Omega_\e$ is a single component, these are joined at vertices $a_\e, b_\e, \dots$ corresponding to the appropriate marked prime ends.  In the multi--component case, if necessary, a similar procedure may be implemented, implying the possible existence of several $a_\e$'s etc.  When required, \emph{the} $a_\e, b_\e, \dots$, etc., will be the one corresponding to the ``principal'' component of $\Omega_\e$, namely, the component which contains the point $z_0$, which, we recall, served to normalize the uniformization map.  Here it is tacitly assumed that $\varepsilon$ is small enough so that this component has a representative of each type.

Further, we require the following:

\vspace{2mm}
(i) It is always the case that $\Omega_\e^\bullet \cup \partial \Omega_\e^\bullet \subset \Omega$.  That is, $\Omega_\e^\bullet$ is in fact a \emph{strictly} inner approximation.

\vspace{2mm}
This property ensures that indeed all of $\overline \Omega_\e$ can be viewed under the (single) conformal map $\varphi: \mathbb D \rightarrow \Omega$ in the ensuing arguments. 

\vspace{2mm}
(ii) Each $z \in \Omega$ lies in $\Omega_\e^\bullet$ for all $\e$ sufficiently small.  

\vspace{2mm}
It can be seen that conditions (i) and (ii) imply that for any $z \in \partial \mathcal A$, there exists some sequence $z_\e \rightarrow z$ with $z_\e \in \mathcal A_\e$, and similarly for $\mathcal B$, etc.

\vspace{2mm}
 (iii) Given any sequence $(z_\e)$ with $z_\e \in \mathcal A_\e$ for all $\e$, any subsequential limit must lie in $\mathcal A$.  Moreover, this must be true in the stronger sense that for any subsequential limit $\varphi^{-1}(z_{\e_n}) \rightarrow \zeta \in \partial \mathbb D$ then $\zeta \in \varphi^{-1}(\mathcal A)$.  Similarly for $\mathcal B$, etc.

\vspace{2mm}
In particular, any subsequential limit of the $(a_\e)$'s will converge to a point in $a$, and similarly for $b$, etc.
\end{defn}

\begin{remark}\label{int_remark}~~~
\vspace{0.2cm}

$\bullet$ To avoid confusion, by the above method, an interior approximation to any slit domain -- no matter how smooth the slit -- necessarily consist of at least a small cavity of a few lattice spacings.  It is noted that the explorer process itself produces just such a cavity. 

$\bullet$ It is easy to check that interior approximations satisfy conditions ($i_I$), ($i_{II}$), ($e$).

$\bullet$ Condition (iii) is indeed used to ensure that the limiting boundary values are unambiguous and correspond to the desired result (see Lemma \ref{int_bound_value}).  A simple scenario where careless approximation leads to the wrong boundary value is illustrated in Figure \ref{bad1}.  

$\bullet$  Note that even though for convenience we have assumed in (iii) that $z_\e \in \mathcal C_\e$ and have used the uniformization map $\varphi$, what is sufficient is that if $z_k \rightarrow z \in \mathcal C$, then for all but finitely many $k$, $z_k$ should be close to $\mathcal C_\e$, in some appropriate sense.  Indeed, we shall have occasion to formulate such a definition later, for the statement of Lemma \ref{top_consist}.

%$\bullet$ We will later have to enlarge the scope of the sort of discretizations we consider which involves relaxation of condition (i) (\mage{see section...}).  The proof of convergence to Cardy's Formula should still hold provided conditions (ii) and (some version of) (iii) can be checked (\mage{c.f.~Lemma???}).
\end{remark}

\begin{figure}
 \vspace{-0.7cm}
\begin{center}
    \includegraphics[width=4.5in]{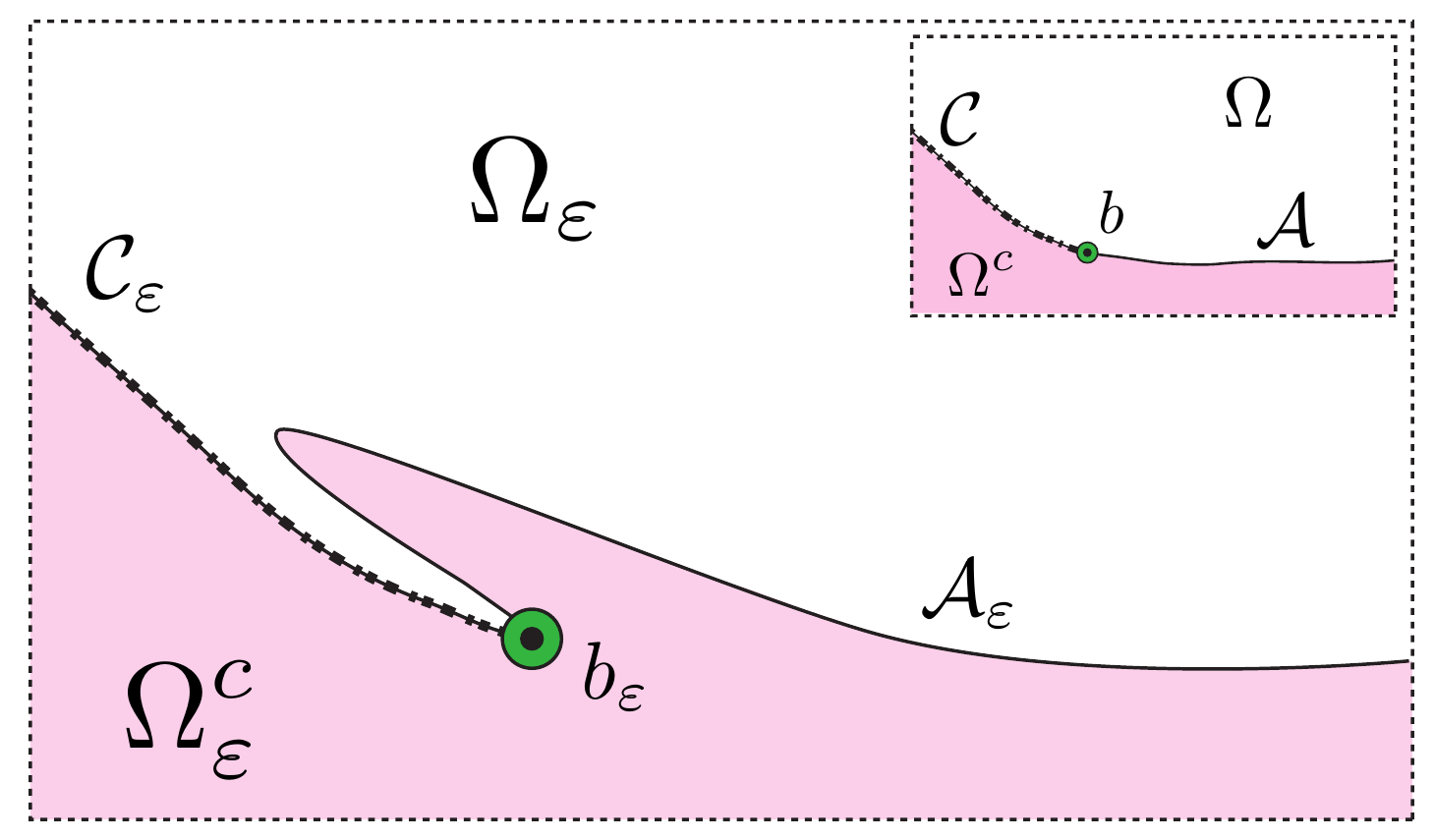}
 \end{center}
 \vspace{-0.5cm}
 \caption{\footnotesize{Violation of condition ii) in Definition \ref{int_app}, which would lead to incorrect (limiting) boundary values.}}\label{bad1}
 \vspace{-0.5cm}
 \end{figure}

\begin{ex}\label{can_app}
An example of an interior approximation is what we will call the \emph{canonical approximation}, constructed as follows.  To be definitive, consider a tiling problem with finitely many types of tiles.  We formally define the scale $\e$ to be the maximum diameter of any tile.  As usual, we may regard all of $\mathbb C$ as having been tiled -- ``$\mathbb C_\e$''.  The domain $\Omega_\e$ is defined as precisely those tiles in $\mathbb C_\e$ which are entirely (including their boundary) in $\Omega$.  Clearly then this construction satisfies (i); condition (ii) is also satisfied: if $z \in \Omega$ is such that $d(z, \partial \Omega) > \e_0$, then $z \in \Omega_\e$ for all $\e < \e_0$.  

At this stage $\partial \Omega_\e$ is just one or more closed polygonal paths.  The boundary component types are determined as follows:  For the marked points, e.g., $a$, consider the neighborhood $Q_\delta(a)$ defined as follows: Let $\mathfrak c_\e$ denote a sequence of crosscuts of $\varphi^{-1}(a)$ with the property that $\varphi(\mathfrak c_\e)$ contains a $\delta$ neighborhood of $a$ with $\delta/\e \rightarrow \infty$ and $\delta(\e) \rightarrow 0$; $Q_\delta(a)$ is then the set bounded by $\varphi(\mathfrak c_\e)$ and the relevant portion of $\partial \Omega$.
  It is clear, for $\e$ small, that ``outside'' these neighborhoods, the assignment of boundary component type is unambiguous.  Here we say a boundary segment is ``outside'' $Q_\delta(a)$ etc.,  if all tiles (intersecting $\Omega$) touching the segment in question lie in the complement of $Q_\delta(a)$.  Indeed, each segment of $\partial \Omega_\e$ belongs to a tile that intersects the boundary.  For a fixed element of $\partial \Omega_\e$ satisfying the above definition of ``outside'', \emph{some} of the external tile is in $\Omega$ and therefore under $\varphi^{-1}$, the image of this portion of the tile joins up with $\partial \mathbb D$; furthermore, it joins with a unique boundary component image due to the size of the obstruction provided by $Q_\delta(a)$.  Finally, inside these neighborhoods $Q_\delta(a)$, etc., all that must be specified are the points $a_\e$, etc., which as discussed above, may have multiple designations (due to the possibility of multiple components for $\Omega_\e$).  The rest of the boundary is then assigned accordingly.
  
Finally, let us establish (iii):
\begin{claim}
The canonical approximation satisfies (iii).
\end{claim}
\begin{proof}
Let $z_\e \in \mathcal A_\e$ with some subsequential limit $z$.  It is clear that $z \notin \Omega$ since all $z \in \Omega$ are a finite distance from the boundary while $d(z_\e, \partial \Omega) \leq \e$ by construction.  Moreover, $z \in \mathcal A$ since $d(z_\e, \mathcal A)$ is (generally less than $\e$ but certainly) no larger than $\delta(\e)$.  It remains to show the stronger statement that any subsequential limit of $\varphi^{-1}(z_\e)$ is in $\varphi^{-1}(\mathcal A)$.  If $\zeta_\e = \varphi^{-1}(z_\e)$ converges to the image of a \emph{marked} point in $\varphi^{-1}(\mathcal A)$ there is nothing to prove.  Thus we may assume that, eventually, $\zeta_\e$ is outside any $\kappa$--neighborhood of the marked points $\alpha_1, \alpha_2 \in \varphi^{-1}(\mathcal A)$ for some $\kappa$.  Now let $\eta < \kappa$ such that the $\eta$ neighborhood of $\partial \mathbb D \setminus [B_\kappa(\alpha_1) \cup B_\kappa(\alpha_2)]$ consist of two disjoint components, one containing all of the rest of $\varphi^{-1}(\mathcal A)$ and the other associated with $\varphi^{-1}(\partial \Omega \setminus \mathcal A)$.  Finally consider the neighborhood (here $\mathcal N_\eta(\cdot)$ denotes the Euclidean $\eta$ neighborhood of $(\cdot)$)
\[ M_\eta := \mathcal N_\eta(\mathcal A) \cap \varphi[\mathcal N_\eta(\varphi^{-1}(\mathcal A)].\]
Since it is agreed that $z_\e$ stays outside $\varphi(B_\kappa(\alpha_1) \cup B_\kappa(\alpha_2))$ it is clear that, for all $\e$ sufficiently small, $z_\e \in M_\eta$ and therefore $\zeta_\e \in \mathcal N_\eta(\varphi^{-1}(\mathcal A)) \setminus [B_\kappa(\alpha_1) \cup B_\kappa(\alpha_2)]$ and not in the complementary $\eta$ band described above.  It follows that the limit must be in $\varphi^{-1}(\mathcal A)$.

%Consider some sequence $z_\e$ with say $z_\e \in \mathcal A_\e$ for all $\e$ and denote $\zeta_\e = \varphi^{-1}(z_\e)$.  Let $g > 0$ and first suppose that the $\zeta_\e$'s lies outside a $g$ neighborhood of the marked prime ends $\{\varphi^{-1}(a), \varphi^{-1}(b), \dots\}$ for all but finitely many $\e$.  Assume now towards a contradiction that some subsequence $\zeta_{\e_n} \rightarrow \zeta \in \varphi^{-1}(\mathcal B)$.  Outside the $g$ neighborhood of $\varphi^{-1}(b)$ (here $b$ is where $\mathcal A$ and $\mathcal B$ join), $d(\varphi^{-1}(\mathcal A) \setminus B_{g}(\varphi^{-1}(b)), \varphi^{-1}(\mathcal B) \setminus B_{g}(\varphi^{-1}(b))) < \eta$ for some $\eta$.  Finally set $\kappa < \eta$ and note that for $n$ sufficiently large, $\zeta_{\e_n} \in B_{\kappa}(\zeta)$ and hence $d(\zeta_{\e_n}, \mathcal B) < \kappa$.  Now $d(z_{\e_n}, \mathcal A) \leq \e_n$ by construction of the canonical discretization and so for $n$ sufficiently large $d(\varphi^{-1}(\mathcal A), \zeta_{\e_n}) < \kappa$, thus contradicting $d(\varphi^{-1}(\mathcal A) \setminus B_{g}(\varphi^{-1}(b)), \varphi^{-1}(\mathcal B) \setminus B_{g}(\varphi^{-1}(b))) < \eta$.  Therefore it must be the case that $\zeta_\e$ lies in $B_g(\varphi^{-1}(b))$ for infinitely many $\e$.  Since $g$ is arbitrary, we conclude that all subsequential limits must be $\varphi^{-1}(b)$.
\end{proof}
\end{ex}

\section{Sup--Approximations}

Unfortunately, for various purposes, e.g., certain proofs of convergence to SLE$_6$, we will need slightly more generality than the internal approximations as provided in Definition \ref{int_app}.  Specifically, we shall have to consider slit domains where (in a certain sense) the slit is evolving dynamically and where, 
at the 
$\varepsilon$--level, the slit is determined stochastically.  In particular, we are not at liberty to approximate the domains in the most convenient fashion; more generality will be required.  

Here, informally, we will describe the two additional properties which are essential in this context:

$\bullet$ Actual sup--norm convergence of separate sides of the slits (which in the discrete approximations may well be separate curves):  This is to prevent the masking of one boundary value by another near the joining of boundaries.

$\bullet$  The \emph{well--organization} property:  This is to prevent confusion of boundary values that could be caused by intermingling (crisscrossing) of the two curves approximating the opposite sides of the slit.

Scenarios in violation of these properties are depicted in Figure \ref{bad2}.

\begin{remark}\label{sup_dist}
If $\gamma_1$ and $\gamma_2$ are two curves, then as usual the sup distance between them is given as  
\[\mbox{dist}(\gamma_1, \gamma_2) = \inf_{\varphi_1, \varphi_2} \sup_t |\gamma_1(\varphi_1(t)) - \gamma_2(\varphi_2(t))|.\]
For certain purposes, it is pertinent to consider weighting the sup--norms of portions of the curves in accord with the particular crosscut in which the portion resides.  We will denote the associated distance by $\mathbf{Dist}$; see \cite{pt1}, $\S3.2$ for the definition and discussions.  However, our ensuing arguments will not be sensitive as to whether we are using the original sup--norm or the weighted version and thus we will continue to use the sup--norm.
\end{remark}

\begin{figure}
 \vspace{-0.7cm}
\begin{center}
    \includegraphics[width=6in]{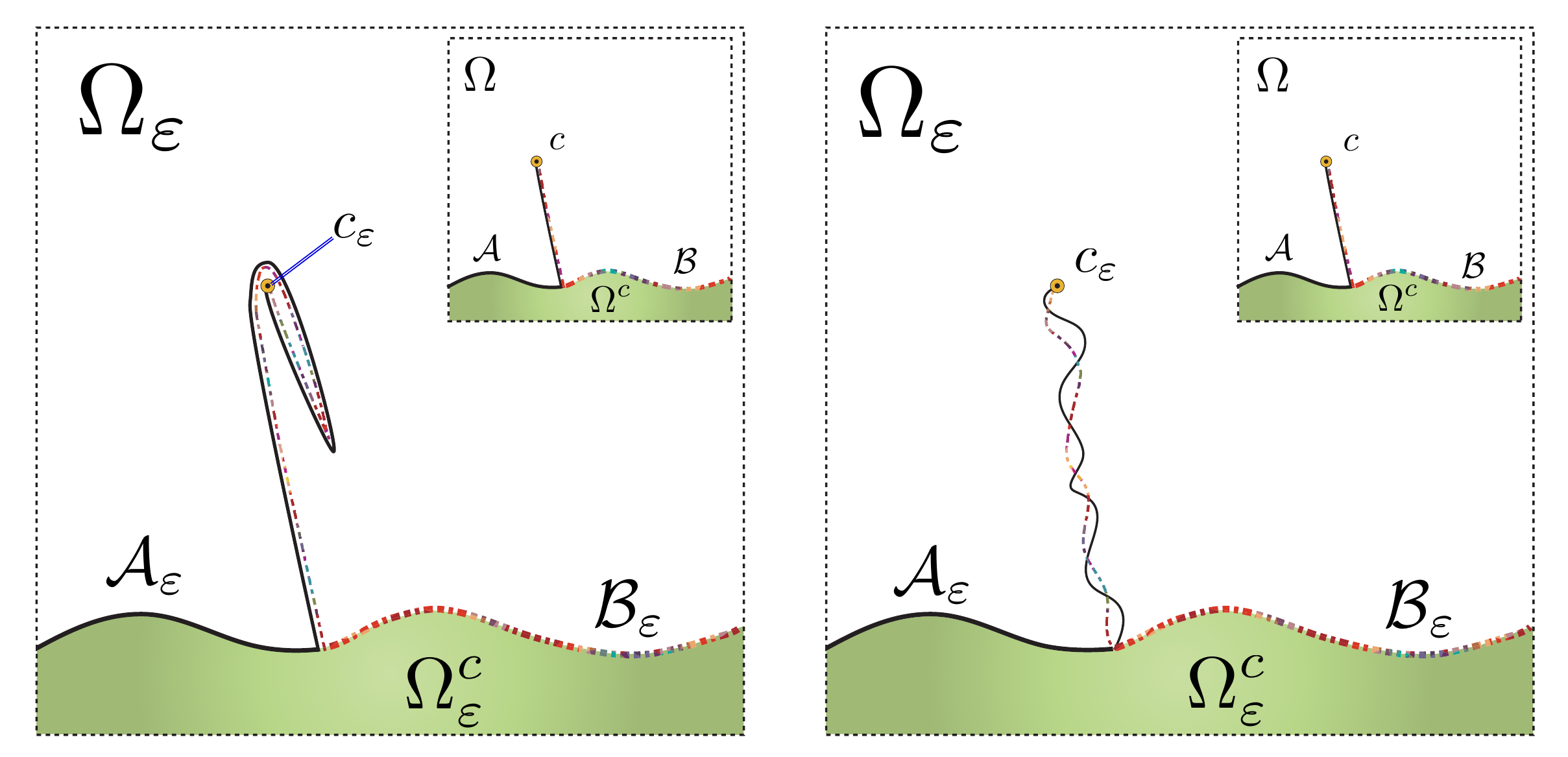}
 \end{center}
 \vspace{-0.5cm}
 \caption{\footnotesize{Masking and intermixing of boundary values.}}\label{bad2}
 \vspace{-0.5cm}
 \end{figure}

\begin{defn}[Sup--approximations]\label{sup_defn}
Suppose $\partial \Omega$ can be further divided (perhaps by other marked boundary points) with the boundary between \emph{these} points described by Jordan arcs or, more generally, 
L\"owner curves.  We shall label the new points $J_1, J_2$, etc. and between certain pairs, e.g.,  $J_k ~\&~  J_{k+1}$  will be a  
L\"owner curve denoted by 
$[J_k, J_{k+1}]$.  The marked prime ends $a, b, \dots$ may serve as an endpoints of (some of) these segments, but it is understood that they do not reside \emph{inside} these arcs.  

Some of this curve (often enough all of it) will be part of the boundary $\partial \Omega$.  (On the other hand, it can be envisioned that a portion of this curve lives in a ``swallowed'' region and is part of $\Omega^c$.)  At the discrete level, we recall that $\partial \Omega_\e$ is automatically a union of closed self--avoiding curves.  It will be supposed that $\Omega_\e$ has corresponding $J_1^\e, J_2^\e, \dots$ and the relevant portion of the curve between the relevant $J$--pair converges in sup--norm to the corresponding portions in $\partial \Omega$ -- or $\Omega^c$, as the case may be -- at rate $\eta(\e)$.  

We assume that all of this transpires in such a way that the following property, which we call \emph{well organized}, holds:
For any curve of interest $[J_k^\e, J_{k+1}^\e]$, pick points $p$ and $p^\prime$ on this arc.  Consider $\delta$--neighborhoods around $p$ and $p^\prime$ and consider the portion of the arc joining these neighborhoods (last exit from neighborhood around $p$ to first entrance to neighborhood around $p^\prime$), which we label $\mathscr L$.  Let $\mathscr P$ be any path connecting the boundaries of these neighborhoods to one another in the complement of all $\partial \Omega_\e$.  Then the relevant portions of the $\partial B_\delta(p)$, $\partial B_\delta(p^\prime)$, $\mathscr P$ and $\mathscr L$ clearly form a Jordan domain, whose interior we denote by $\mathcal O$.  Let $\mathcal O^\prime \subset \mathcal O$ denote the connected component of $\mathscr P$ in $\overline{\mathcal O} \setminus \partial \Omega_\e$.  Then, $\partial \mathcal O^\prime \cap \partial \Omega_\e$ is \emph{monochrome}, i.e., it cannot intersect both $[J_k^\e, J_{k+1}^\e]$ and $[J_\ell^\e, J_{\ell+1}^\e]$ for $k \neq \ell$.  While this may sound overly complicated, what we have in mind is actually a simple topological criterion, c.f., Remark \ref{sup_remark}.

The rest of the domain and boundary is approximated by interior approximation.  Thus, for those $J_k$'s which divide arc--portions of $\partial \Omega$ from ``other'', we require \emph{commensurability} at the joining points.  In particular, in order that the interior approximation be implementable, it is clear that we must require 
$J_k^{\e}\in \Omega$.
\end{defn}

\begin{remark}\label{sup_remark}~~~
\vspace{0.2cm}

$\bullet$ While at first glance it is difficult to imagine that $\partial\mathcal O^\prime$ is anything except, say $\mathscr L$, what we have in mind is when $\mathscr L$ and a neighboring curve are some approximation to a two--sided slit.  The well--organized property does not permit the sides of the approximation to crisscross one another.  Alternatively, this is a simple topological criterion which can be phrased as saying that under say the uniformization map (in fact any homeomorphism onto a Jordan domain would do) the image of each of these $J$--pieces occupies a single contiguous piece of the boundary.  This sort of monochromicity property is required for well--behaved convergence of relevant boundary conditions we shall need later.  
%E.g., suppose $\partial \Omega$ consist of a slit which should have boundary value (for some function) which is one on the ``blue'' side and zero on the ``yellow'' side.  
Then a crisscrossing approximation can very well lead to altogether different limiting values -- or none at all.  It is clear that this well--organized property is satisfied by the trace of any discrete percolation explorer process.

$\bullet$ Sup--approximations satisfy conditions ($i_I$), ($i_{II}$), $(e)$.

$\bullet$ The added difficulty here is that since the approximation is no longer interior, we can no longer determine the ``topological situation'' by looking under a single conformal map.  E.g., for a point close to the boundary, we can no longer determine which boundary piece it is ``really'' close to.  This is exemplified by the case of a slit domain: If, say, part of $\mathcal C$ is one side of a two--sided slit $\gamma$, then points close to $\gamma$ on one side (corresponding to $\mathcal C$) will have small $u$ value which tends to 0 whereas points close to $\gamma$ on the other side (corresponding to say $\mathcal B$) will tend to non--trivial boundary values.  In the case of interior approximation all such ambiguities were resolved by looking under the conformal map $\varphi^{-1}$.  

$\bullet$ It is worth noting that the important case in point where the boundary consist of an original $\Omega$ with a (L\"oewner) slit -- which might be two sided -- falls into the setting under consideration.  In particular, we will have occasion to consider cases where we have $\gamma_n \rightarrow \gamma$ in the sup--norm with $\gamma_n$ being discrete explorer paths.  In this case to check condition ($i_{II}$), we observe that if $\gamma_n \rightarrow \gamma$ in sup--norm and $z_n \in \gamma_n$ and $z_n \rightarrow z$, then $z \in \gamma$ and hence certainly in the complement of the the domain of interest $\Omega \setminus \mathbb I(\gamma)$ (i.e., $\Omega$ delete $\gamma$ together with components ``swallowed'' by $\gamma$).  For an illustration see Figure \ref{swallow}.  These circumstances may be readily approximated by a hybrid of sup-- and canonical approximations and, as is not hard to see, satisfy the condition of commensurability. 
\end{remark}

\begin{wrapfigure}{r}{0.4 \textwidth}
 \vspace{-0.84cm}
\begin{center}
    \includegraphics[height=2.2in]{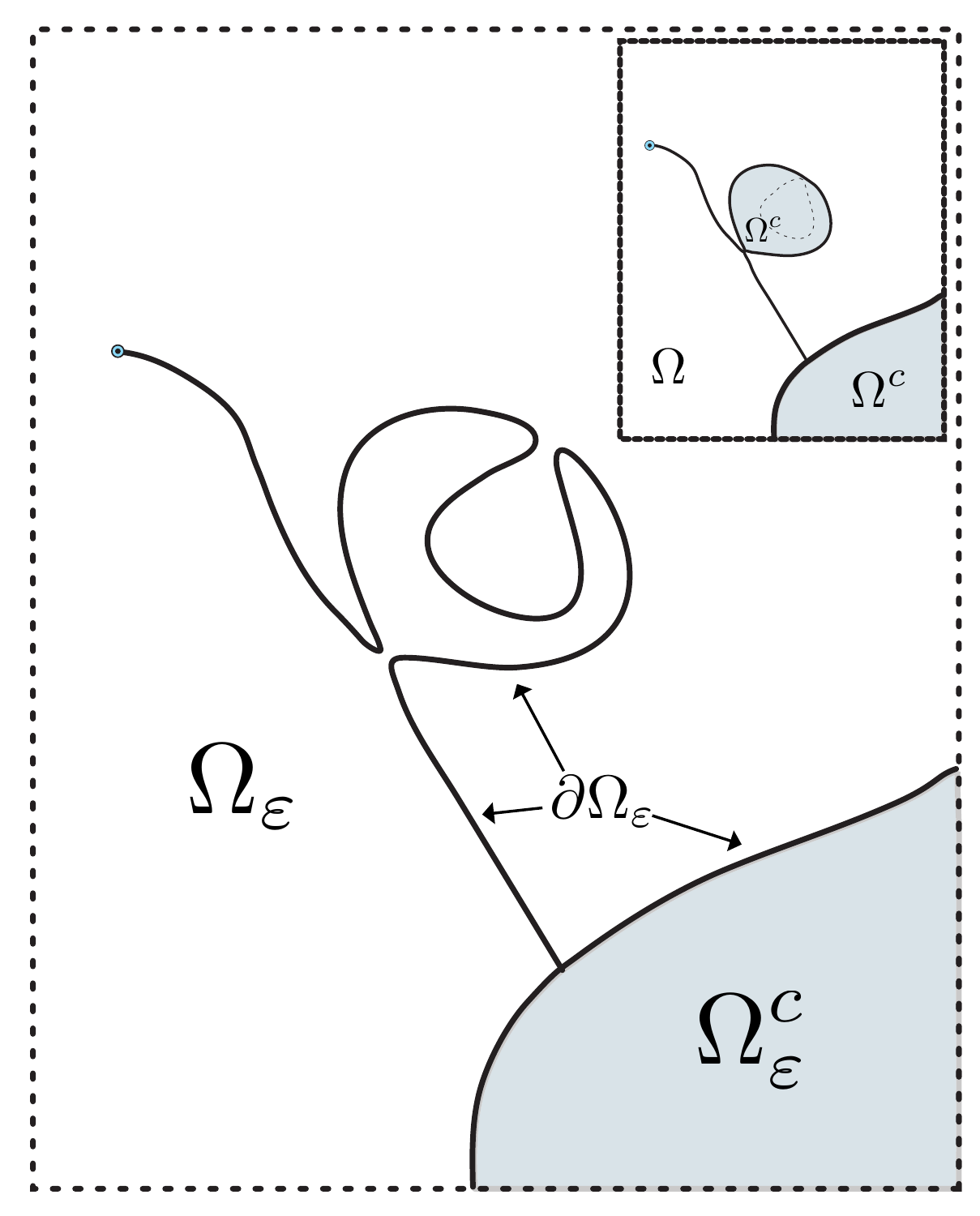}
 \end{center}
 \vspace{-0.5cm}
 \caption{\footnotesize{A case where the limiting domains does not contain a component present in approximating domains.  Due to frequent self--touching, such (limiting) domains are in fact typical of SLE$_6$.}}\label{swallow}
 \vspace{-0.3cm}
 \end{wrapfigure}

The main addition is the following lemma which serves the role of condition (iii) in Definition \ref{int_app} to ensure unambiguous retrieval of boundary values (see Lemma \ref{sup_bound_value}).  
That is, if $w\in \Omega$ is close to $\mathcal C$ in the ``homotopical sense'' that any \emph{short} walk from $w$ which hits $\partial \Omega$ must hit the $\mathcal C$ portion of $\partial \Omega$ then $w$ is close to $\mathcal C_n$ by the same criterion.  A precise statement of this intuitive notion is, unfortunately, much more involved.

%That is, if $w$ is close to $\mathcal C$, then it is ``topologically'' close to $\mathcal C_n$ (and not a different boundary component) where $\mathcal C_n$ approximates $\mathcal C$, in the sense that any short walk from $w$ which hits $\partial \Omega_n$ must hit $\mathcal C_n$.

\begin{lemma}[Homotopical Consistency]\label{top_consist}
Consider a domain $\Omega$ with marked boundary prime ends $a, b \in \partial \Omega$.  Let us focus on boundary $\mathcal C$ with end points $a$ and $b$ which we consider to be the bottom of the boundary.  (Note that $\mathcal C$ may consist of Jordan arcs together with arbitrary parts -- if double--sided slits are involved, such that not both sides belong to $\mathcal C$, then the corresponding arc(s) must be connected all the way up to $b$ and/or $a$).  Let us denote the sup--approximation to $\Omega$ by $\Omega_n$ and the portion of the boundary approximating $\mathcal C$ by $\mathcal C_n$.  

%We also write $d(\mathcal C_n, \mathcal C)$ for the distance between them, where the distance is measured in the sup--norm where appropriate and the Hausdorff metric otherwise.  

Suppose we have a point $q$ which is more than $\Delta$ away from $a$ and $b$ and $\delta^\star$ away from $\mathcal C$ with $\Delta \gg \delta^\star$, such that $\vartheta = \varphi^{-1}(q)$ is close to $\varphi^{-1}(\mathcal C)$.  Then there exists $\eta > 0$ with $\eta \ll \delta^\star$ such that if $\mbox{\emph{dist}}(\mathcal C_n, \mathcal C) < \eta$ (here \emph{dist} denotes e.g., the sup--norm distance where appropriate, and otherwise the Hausdorff distance) then there exists some path $\mathscr P$ from (some point in) $\varphi^{-1}(B_\Delta(a))$ to (some point in) $\varphi^{-1}(B_\Delta(b))$ (we denote this by $\varphi^{-1}(B_\Delta(a)) \leadsto \varphi^{-1}(B_\Delta(b))$) such that in the sup--approximation $\Omega_n$, $q$ is in the bottom component of $\Omega_n \setminus \varphi(\mathscr P \cup B_\Delta(a) \cup B_\Delta(b))$ and further, any walk from $q$ in the bottom component which hits $\partial \Omega_n$ must hit $\mathcal C_n$.
\end{lemma}

\begin{proof}
For clarity, we divide the proof into four parts.  

1.  We let $\eta \ll \Delta \ll 1$ and consider, under the uniformization map, the set 
\[\mathscr B := \mathbb D \setminus [\varphi^{-1}(B_\Delta(a)) \cup \varphi^{-1}(B_\Delta(b))].\]
Let us now draw a path $\mathscr P^\prime: \partial \varphi^{-1}(B_\Delta(a)) \leadsto \partial \varphi^{-1}(B_\Delta(b))$ which defines top and bottom components in $\mathscr B$ with $\omega$ in the bottom component, and hence also the bottom component of $\varphi(\mathscr B)\setminus \varphi(\mathscr P^\prime)$). Further, $\mathscr P := \varphi(\mathscr P^\prime)$ is some finite distance $\delta \gg \eta > 0$ away from $q$.  (In  essence, $\delta$ will now play the role of $\delta^\star$ in the statement of the Lemma.)

\begin{figure}
\begin{center}
    \scalebox{0.47}{\input{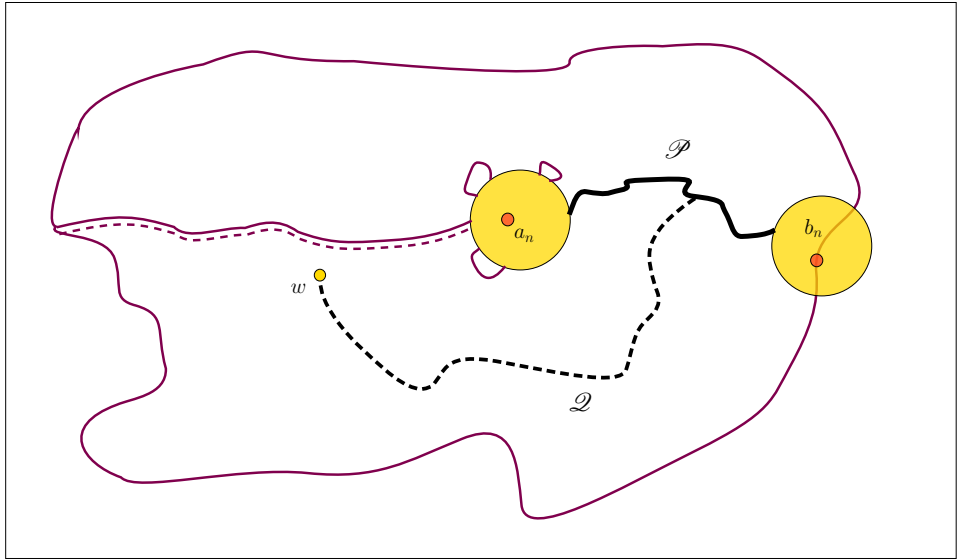}}
 \end{center}
 \caption{\footnotesize{The domain $\mathscr V_n$, etc.}}\label{top_consist1}
 \end{figure}

2.  We now look at the domain 
\[ \mathscr V_n = \Omega_n \setminus [B_\Delta(a) \cup B_\Delta(b) \cup \mathscr P].\]
We claim that for $n$ sufficiently large, all of the above is well--defined: Indeed, $\mathscr P$ is a compact set in $\Omega$ and hence for $n$ sufficiently large, is contained in $\Omega_n$, by Proposition \ref{metric_cond}.    Of course, $\Omega_n$ itself may have many components; we are focusing on the principal component.  Even so, with the above setup, $\mathscr V_n$ may also have many components, e.g., near the boundaries of $B_\Delta(a)$ and $B_\Delta(b)$ (see Figure \ref{top_consist1}).

However, we claim that it has the analogue of a top and bottom component: Indeed, it is clear that ``large'' compact sets in $\mathscr B$ well away form the boundaries continue to lie in large connected components of $\mathscr V_n$.  More quantitatively, while at the scale $\Delta$, $\partial \Omega_n$ may create various components by entering and re--entering $B_\Delta(a)$ and $B_\Delta(b)$, since $\mathcal C_n$ and $\mathcal C$ are $\eta$--close (say in the Hausdorff distance), if we shrink these neighborhood balls to scale $\Delta - 2\eta$, then such components merge into (the) two principal components, leaving only $\eta$--scale small components in the vicinity of the neighborhood balls.  

Finally, we claim that $q$ is in the bottom component of $\mathscr V_n$.  First, since $B_\delta(q)$ must all be in the same component of $\mathscr V_n$, $q$ cannot be in a small $\eta$--scale component.  
The argument can be finished by any number of means.  For example we may choose to regard $\mathscr P$ 
as two--sided; the component of $q$ is determined by which side of $\mathscr P$ it may be connected to.   For future reference, 
let $\mathscr Q^\prime \subset \mathbb D$ denote a simple path (staying well away from $\partial \mathscr B$) connecting
$\vartheta$ to $\mathscr P^\prime$ and 
$\mathscr Q$ the image of $\mathscr Q^\prime$ under
$\varphi$.  Then, again,  by Proposition \ref{metric_cond}, for all $n$ sufficiently large, the entirety of $\mathscr Q$ is found in $\Omega_n$ and the appropriate component -- bottom -- for $q$ is determined for once and all.  The relevant domains, etc., are illustrated in Figure \ref{top_consist1}.

%
%
%\vspace {2 cm}
%
%
%
%To finish, we note that the same sort of description of components, etc., for $\Omega_n$ also holds in the intersected domain $\mathscr V_n \cap \Omega$ and hence $w$ must be in either the top or the bottom component of $\mathscr V_n \cap \Omega$, and since we already know that $w$ is in the bottom component of $\varphi(\mathscr B) \setminus \varphi(\mathscr P^\prime)$, it must also be in the bottom component of $\mathscr V_n$.  \mage{this needs to be clarified, etc., L. not happy}

3.  It is clear that $q$ is close to $\mathcal C_n$.  We further claim that it is not obstructed from $\mathcal C_n$ by other portions of $\partial \Omega_n$, as may be the worry when a portion of $\mathcal C$ is (one side of) a two--sided slit.  We need to divide into a few cases.  First if the only portion of $\mathcal C$ which is close to $q$ is approximated interiorly, then by an investigation of the situation under the uniformization map, it is clear that no obstruction is possible.  So now we suppose that $q$ is close to some Jordan arc $\mathscr J: = [J_k, J_{k+1}]$.  If $\mathscr J$ is one--sided, then there is no problem, since then $q$ is not close in anyway to any other portion of the boundary except near the endpoints, which we may assume, by shrinking relevant scales if necessary, that $q$ is far away from.  

4.  We are down to the main issue where $\mathscr J$ is a two--sided slit, which is being sup--norm approximated by $\mathscr J_n$.  Since at least one of the end points must be $a$ or $b$, let us assume without loss of generality that $J_k = a$.  We will need to do some refurbishing, starting with the neighborhood balls around $a$ (and $b$, if necessary).  Let $\mathfrak q$ be a point on $\mathscr J$ near $a$.  It is manifestly the case that $\mathfrak q$ has two images under $\varphi^{-1}$ -- which are near $\varphi^{-1}(a)$; consider a crosscut between these two images; the image of this crosscut under $\varphi$ then defines the relevant neighborhood, which we will denote by e.g., $B(a)$.  We note that i) by construction, $B(a)$ has the property that $\mathscr J$ enters exactly once and terminates at $a$, and ii) being slightly more careful if necessary to ensure the relevant crosscut is contained in $\varphi^{-1}(B_\Delta(a))$, we can also ensure that $B(a) \subset B_\Delta(a))$.  Here we will consider $\eta \ll \mbox{dist}(a, \partial B(a))$, so that in particular, e.g., $a_n \in B(a)$.

Now let us return attention to $\Omega_n$.   We will now refurbish $\mathscr P$ so that it directly joins $a_n$ to $b_n$ and avoids all of $\partial \Omega_n$; we will call the resultant path $\mathscr P_r$. We claim that it is possible to draw such a $\mathscr P_r$ by suitably extending $\mathscr P$, under the above stipulations concerning $B(a)$, $B(b)$, and $\eta$.  Focusing attention on $B(a)$, if this were not possible, then it must have been the case that a portion of $\mathscr J_n$ or a portion of $\partial \Omega_n \setminus \mathscr J_n$ which is approximating the other side of $\mathscr J$, is obstructing.  This scenario implies an inner domain inside $B(a)$ surrounding the tip $a_n$ with boundary e.g., $\mathscr J_n$.  Since $\eta \ll \mbox{dist}(a, \partial B(a))$, this violates sup--norm $\eta$ closeness.  (Here it appears that the sup--norm closeness property is crucial.  For an illustration see Figure \ref{top_consist2}.)  

\begin{figure}
 \vspace{-0.7cm}
 \hspace{0.2cm}
\begin{center}
    \scalebox{0.35}{\input{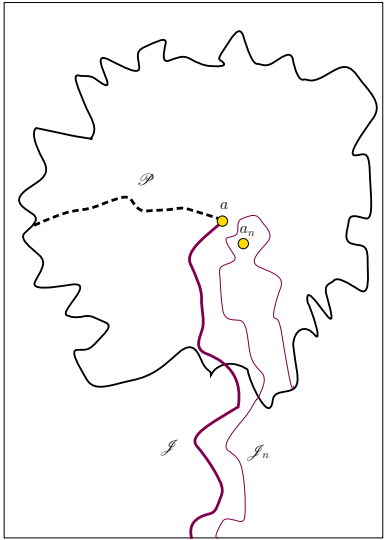}}
 \end{center}
\vspace{-0.2cm}
 \caption{\footnotesize{Failure to continue $\mathscr P$ to $\mathscr P_r$ inside $B(a)$.}}
\label{top_consist2}
 \end{figure}

Having achieved all this, it is again clear that the principal component of $\Omega_n$ is divided into two disjoint Jordan domains.  Indeed by the fact that the approximation is well--organized, there are two circuits -- both using $\mathscr P_r$, passing through $a_n$ and $b_n$, such that one (which again is the bottom one) contains $\mathcal C_n$ and the other contains (the principal component of) $\Omega_n \setminus \mathcal C_n$, with no possibility of mixing via crisscrossing.  Since $\mathscr P_r$ is an extension of $\mathscr P$, it is clear from the closing argument of 3) that $q$ is in the bottom component and hence must be closed to $\mathcal C_n$ without obstruction from any portion of $\partial \Omega_n \setminus \mathcal C_n$.  

\end{proof}

\section{Verification of Boundary Values for $u, v, w$}
We are now in a position to verify boundary values for $u, v, w$ using RSW estimates.  Let us begin with a more detailed recapitulation/clarification of how we take the scaling limit of $u_\e, v_\e, w_\e$ (see \cite{stas_perc} and \cite{cardy}).  Consider some exhaustion $K_n \nearrow \Omega$, with $K_n$ compact.  The RSW estimates imply equi--continuity, and hence we have $u_{\e_k}^{(n)} \rightarrow u^{(n)}$ uniformly on $K_n$ and, at least for the models in \cite{stas_perc} and \cite{cardy}, $$F^{(n)}:=u^{(n)} + e^{2\pi i/3}v^{(n)} + e^{-2\pi i/3} w^{(n)}$$ is analytic there with
$$ u^{(n)} + v^{(n)} + w^{(n)} = \mbox{const}.$$
We may take $(\e_k^{(n+1)}) \subset (\e_k^{(n)})$ as a subsequence which implies that 
$$
u_{\e_k}^{(n+1)} \rightarrow u^{(n+1)}
$$
etc., so that $F^{(n+1)}$ is analytic in $K_{n+1}$ with values agreeing with the old $F^{(n)}$ in $K_n$.  The diagonal sequence $(u_{\e_n}^{(n)})$ converges uniformly on compact sets to some $u$; together with similar statements for $v$ and $w$, we obtain that the limiting $F$ is analytic on $\Omega$.  In the sequel for simplicity we will drop the $(n)$ superscripts and e.g., simply denote $u_\e \rightarrow u$. 

We begin with a lemma which provides us with the RSW technology which is necessary for establishing boundary values.

\begin{lemma}\label{boundary_annuli}
Let $\Omega$ and $\varphi$ be as described. The pre--image of $\partial \Omega$ under $\varphi$ is divided into a finite number of disjoint (connected) closed arcs the intersection of any adjacent pair of which is the corresponding (pre--image of the) prime end.  Then for $z \in \partial \Omega \setminus \{a, b, c, \dots\}$, we identify $z$ with a single $\varphi^{-1}(z)$ and similarly identify its corresponding boundary component.

(I) There exists an infinite sequence of (``square'') neighborhoods $(S_\ell)$ centered at $z$ such that $S_\ell \cap \Omega \neq \emptyset$ for all $\ell$ and $S_{\ell+1}$ is strictly contained in $S_\ell$ with $\partial S_\ell$ containing portions of the boundary component containing $z$ and 

(II) In each $S_\ell \setminus S_{\ell+1}$, there is a ``yellow'' circuit and/or a ``blue'' circuit which separates $z$ from all other boundary components with probability that is uniformly positive as $\varepsilon \rightarrow 0$ (provided that $\varepsilon$ is sufficiently small depending on $\ell$).  By separation it is meant that in the pre--image in $\mathbb D$, $\zeta$ is separated from all other boundary components along any path in $\mathbb D$ whose image under $\varphi$ tends to $z$.

Finally, for $z \in \{a, b, c, \dots \}$, a similar statement holds, except for the fact that here the relevant circuits separate $z$ from all other boundary points and boundary components to which $z$ does not belong.
\end{lemma}

\begin{proof}
Let $z \in \partial \Omega$ and let $\zeta = \varphi^{-1}(z)$ denote its corresponding pre--image.  First suppose $z \notin \{a, b, c, \dots\}$ so that $\zeta$ is some finite distance from the corresponding points on $\partial \mathbb D$.  Next we consider a sufficiently small crosscut $\Gamma$ of $\mathbb D$ surrounding $\zeta$ (a finite distance away from $\zeta$) whose end points on $\partial \mathbb D$, denoted $\alpha$ and $\beta$, are such that $\alpha$ and $ \beta$ are in the (interior of the) boundary component of $\zeta$.  We also denote by $Q$ the image of the interior of the region bounded by $\Gamma$ and the relevant portion of $\partial \mathbb D$; we note that $z \in \partial Q$.  Let $S_0 \subset \mathbb C$ be a small square centered at $z$ whose intersection with $\Omega$ lies entirely inside $Q$.  Then by construction, $\partial (S_0 \cap \Omega)$ can contain at most boundary pieces from the boundary component of $z$.  Now the sequence $S_n$ will be constructed similarly, with the stipulation that the linear scale of $S_{\ell+1}$ is reduced by half. 

By standard RSW estimates for the percolation problem in all of $\mathbb C$, there is a blue and/or yellow Harris ring inside each annulus $S_\ell \setminus S_{\ell+1}$ with probability uniformly bounded from below for $\varepsilon$ sufficiently small (depending on $n$).  Now consider any path $\mathscr P$ in $\mathbb D$ which originates at $\zeta$ and ends outside $\varphi^{-1}(Q)$ such that the image of the path originates at $z$.  Such a path stays in $\Omega$ and therefore must intersect the said circuit.  

Identical arguments hold for $z \in \{a, b, c, \dots\}$ except for the fact that the original crosscut will now originate and end on two distinct boundary components.  
\end{proof}

It is noted that in the presence of a circuit in $S_\ell \setminus S_{\ell+1}$, the above separation argument also applies to points in $\partial (S_m \cap Q)$ if $m \geq \ell + 1$.

\begin{lemma}[Establishment of Boundary Values for Interior Approximations]\label{int_bound_value} Let $\Omega$ and $\varphi$ be as described.  We recall that $u_\e^B(z)$ is the probability at the $\e$ level that there is a blue crossing from $\mathcal A_\e$ to $\mathcal B_\e$, separating $z$ from $\mathcal C_\e$, and let $u$ denote the limiting function.  Then $u= 0$ on $\mathcal C$ in the sense that if $z_k \rightarrow z \in \mathcal C$ in such a way that $\varphi^{-1}(z_k) = \zeta_k \rightarrow \zeta \in \varphi^{-1}(\mathcal C)$, then $\lim_{k \rightarrow \infty} u(z_k) = 0$.  Similarly, in the vicinity of the point $c$, $u$ tends to one.  Analogous statements hold for $v^B_\varepsilon$ and $w^B_\varepsilon$ and for the yellow versions of these functions.
\end{lemma}
\begin{proof}
Suppose a yellow Harris circuit has occurred in $S_\ell \setminus S_{\ell+1}$ and let $z_k \rightarrow z$ as described.  Then, in the language of the proof of Lemma \ref{boundary_annuli}, for $k$ sufficiently large $z_k \in Q \cap S_m$ for some $m = m(k)$ tending to $\infty$ as $k \rightarrow \infty$.  For $\e$ sufficiently small, it follows from (iii) in Definition \ref{int_app} that $\mathcal A_\e$ and $\mathcal B_\e$ are disjoint from $(S_m \cap Q \cap \Omega_\e^\bullet)$ and since $\Omega_\e^\bullet$ is an interior approximation, the relevant portion of the circuit evidently joins with $\partial \mathcal C_\e$ to separate $\partial S_m \cap Q \cap \Omega_\e^\bullet$ from $c$, as for $z$ as discussed near the end of the proof of Lemma \ref{boundary_annuli}.  This separation would preclude the crossing event corresponding to $u_\e^B(z_k)$ since -- as is clear if we look on the unit disc via the conformal map $\varphi^{-1}$ -- the latter necessitates (two) blue connections between the relevant portions of $\partial S_m$ and other boundaries.  Now consider $k$ with $m(k)$ very large; then for all $\e$ sufficiently small, the probability of at least one yellow circuit is, uniformly (in $\e$), close to some $p(m)$ where $p(m) \rightarrow 1$ as $m \rightarrow \infty$.  It therefore follows that $u(z_k) \leq 1-p(m(k)) \rightarrow 0$ as $z_k \rightarrow z$.  Finally, boundary value of $c$ follows the same argument: Here the blue Harris ring events accomplish the required connection between $\mathcal A_\e$ and $\mathcal B_\e$.  Arguments for other functions/boundaries are identical.

\end{proof}

\begin{lemma}
[Establishment of Boundary Values for Sup--approximations]\label{sup_bound_value} Let $\Omega$ and $\varphi$ be as described. We recall that $u_\e^B(z)$ is the probability at the $\e$ level that there is a blue crossing from $\mathcal A_\e$ to $\mathcal B_\e$, separating $z$ from $\mathcal C_\e$, and let $u$ denote the limiting function.  Then $u= 0$ on $\mathcal C$ in the sense that if $z_k \rightarrow z \in \mathcal C$ in such a way that $\varphi^{-1}(z_k) = \zeta_k \rightarrow \zeta \in \varphi^{-1}(\mathcal C)$, then $\lim_{\e \rightarrow 0} u_\e^B(z_k) \rightarrow 0$.  Similarly, in the vicinity of the point $c$, $u$ tends to one.  Analogous statements hold for $v^B_\varepsilon$ and $w^B_\varepsilon$.
\end{lemma}
\begin{proof}
We recall that we have three boundary pieces $\mathcal A, \mathcal B, \mathcal C$ in counterclockwise order, where we assume without loss of generality that $\mathcal C$ is on the bottom.  We also label the relevant marked prime ends $a, b, c$, in counterclockwise order, such that e.g., $c$ is opposite to $\mathcal C$.  Thus if we draw a path $\mathscr P$ between $a_\e$ and $b_\e$ inside $\Omega_\e$, and $z_k \in \Omega_\e$ is inside the region formed by $\mathscr P$ and $\mathcal C_\e$ then to prevent events which contribute to $u^B_\e$, it is sufficient to seal $z_k$ off from $c_\e$ by a yellow Harris ring together with the bottom boundary $\mathcal C_\e$.  This is precisely the setting of Lemma \ref{top_consist} (with $w = z_k$) and so we conclude that for $k$ sufficiently large, for $\e$ sufficiently small depending on $k$, in order to prevent events contributing to $u_\e^B$, it is indeed sufficient to seal $z_k$ off with a yellow Harris ring.

We are now in a position to invoke Lemma \ref{boundary_annuli}.  The proof follows closely as in the last part of the proof of Lemma \ref{int_bound_value} except for one difference:  For $k$ sufficiently large, $z_k \in S_m$ for some $m=m(k)$ which increases as $k$ increases; however, in the case of sup--approximation, it is no longer quite so automatic that arbitrarily small Harris rings will hit the boundary $\mathcal C_\e$. However, given $\e$, we have that $\mathcal C_\e$ is at most a distance $\eta(\e)$ from $\mathcal C$, and thus, for fixed $\e_0$, there is some $M(\e_0)$ such that $m(k) \nearrow M(\e_0)$ as $k \rightarrow \infty$ (and $M(\e_0)) \rightarrow \infty$ as $\e_0 \rightarrow 0$).  So we still have that uniformly for all $\e \leq \e_0$, $U_\e(z) \leq 1 - p(M(\e_0))$, where $p(M(\e_0))$ as before denotes the probability of at least one yellow Harris ring in the annulus $S_1 \setminus S_{M(\e_0)}$, and tends to 1 as $M(\e_0)$ tends to infinity. 

%1.  Show that from previous lemma, enough to seal $z_k$ off from $c$ by Harris ring and parts of $\mathcal C$.  
%
%2.  Invoke Lemma 3.1; for inner approximation part certainly enough; for sup--approximation part, should also be enough since given $\eta$, Harris rings ``large'' enough should do the job, even if $\gamma_n$ leaks outside $\Omega$.

\end{proof}

\begin{remark}\label{cont_to_boundary}
Our arguments in fact show that the function $u$ is continuous up to the boundary: Given any sequence $z_k \rightarrow z \in \mathcal C$, we have that given any $\kappa > 0$, for $k$ sufficiently large, $|u_{\e_n}^{(n)}(z_k)|< \kappa$, uniformly in $n$, for $n$ sufficiently large (or $\e$ sufficiently small) and hence $u(z_k) < \kappa$ (c.f., the end of the proof of Theorem \ref{bv_prop}).  We have similar statements for $v$ and $w$ on the corresponding boundaries.  
\end{remark}

To check that $F$ is indeed the appropriate conformal map and thereby uniquely determine it and retrieve Cardy's Formula, we follow the arguments in \cite{beffara}.  We remark that while there exists certain literature on discrete complex analysis (see e.g., \cite{ferrand} and \cite{stas_chelkak} and references therein) our situation is less straightforward since e.g., none of the functions $u_N, v_N, w_N$ are actually discrete harmonic.  Moreover, due to the fact that we are considering general domains (versus Jordan domains) and $\partial \Omega$ may not be so well--behaved, to obtain conformality requires some extra work.  In any case, we will now amalgamate all ingredients to prove the following result:

\begin{thm}\label{bv_prop}
For the models described in \cite{cardy} (which includes the triangular site problem studied in \cite{stas_perc}), consider the function $F = u + e^{2\pi i/3} v + e^{-2\pi i /3} w$, where $u, v, w$ are the limits of $u_\e, v_\e, w_\e$.  Then $F$ is the unique conformal map between $\Omega$ and the equilateral triangle $\mathbf T$ with vertices at 1, $e^{2\pi i /3}$, $e^{-2\pi i / 3}$.
\end{thm}
\begin{proof}
We claim that the following seven conditions hold:
\begin{enumerate}
\item $F$ is nonconstant and analytic in $\Omega$,
\item $u, v, w$ (and hence $F$) can be continued (continuously) to $\partial \Omega$,
\item $u + v + w$ is a constant,
\item $u(c) = 1$, with similar statements for $v$ and $w$ at $a$ and $b$,
\item $u \equiv 0$ on $\mathcal C$ with similar statements for $v$ and $w$ on $\mathcal A$ and $\mathcal B$, 
\item $F\circ \varphi$ maps $\partial \mathbb D$ bijectively onto $\partial \mathbf T$,
%\item $u \circ \varphi$ is injective on $\varphi^{-1}(\mathcal A)$, with corresponding statements for the other boundaries and other functions and 
\item $(F \circ \varphi) (\mathbb D) \cap (F\circ \varphi) (\partial \mathbb D) = \emptyset$;
\end{enumerate}
from which the proposition follows immediately.  Indeed, from conditions 7 and 6, $F \circ \varphi: \mathbb D \rightarrow \mathbf T$ is a conformal map (this follows directly from e.g., Theorem 4.3 in \cite{lang}).  But clearly, conditions 5, 4, 3 imply that $F$ maps $\Omega$ into $\mathbf T$, and further, conditions 2 and 1 imply that $F$ maps $\Omega$ \emph{onto} $\mathbf T$ (this follows from e.g., Theorem 4.1 in \cite{lang}).  Altogether, conformality of $F$ itself now follows: It is enough to show that $F^\prime$ never vanishes, but this follows from the fact that $0 \neq (F\circ \varphi)^\prime(z) = F^\prime(\varphi(z)) \varphi^\prime(z)$. 

We now turn to the task of verifying conditions 1 -- 7.  It follows from \cite{stas_perc}, \cite{cardy}, and \cite{beffara} that $F$ is analytic and that $u + v + w$ is constant.  On this basis, the real part of $F$ is proportional to $u$ plus a constant and it is seen from Lemma \ref{int_bound_value} (or Lemma \ref{sup_bound_value}) that 
$u$ is {\emph not} constant, i.e., it is close to 1 near $c$ and close to 0 near $\mathcal C$.  We have conditions 1 and 3.  Conditions 2, 4, 5 follow from Lemma \ref{int_bound_value} (or Lemma \ref{sup_bound_value}) and Remark \ref{cont_to_boundary}.

To demonstrate condition 7, let us write $\mbox{Re}(F) = (3/2)u - 1/2$.  Then if we show that $u \neq 0$ in $\Omega$, then we have demonstrated that $F(\Omega)$ does not intersect $F(\mathcal C)$.  The latter follows since once $z \in \Omega$, we can construct a tube of bounded conformal modulus connecting $\mathcal A$ to $\mathcal B$ going underneath $z$, and within this tube, by standard percolation arguments which go back to \cite{ACCFR}, we can construct a monochrome path separating $z$ from $\mathcal C$.  Condition 6 follows in a similar spirit: E.g., on the $\mathcal A$ boundary, if $z \neq q$, but $|z - q| \ll 1$, then by the argument of Lemma \ref{boundary_annuli}, $u(z)$ is close to $u(q)$ (since both can be surrounded by many annuli in which e.g., a blue circuit occurs).  Similar arguments for $v$ and $w$ and other boundaries directly imply continuity of all functions on all boundaries of $\Omega$.  Moreover, this implies, e.g., $u\circ\varphi^{-1}(\mathcal A)$ is continuous on the relevant portion of the circle starting (at 
$\varphi^{-1}(c)$) with the value 1 and ending
(at 
$\varphi^{-1}(b)$) with the value 0 and thus achieving all values in $[0,1]$.  Similarly statements hold for the other functions on the other boundaries.  Condition 6 now follows directly.

\end{proof}
\begin{remark}\label{bv_remark}
It is worth noting that while using only arguments involving RSW bounds, we have determined that 1) the $u, v, w$'s can be continued to the boundary and 2) partial boundary values, e.g., $u \equiv 0$ on $\mathcal C$, sufficient determination of boundary values requires additional ingredients.  In particular, 
 we also needed that e.g., $v + w \equiv 1$ on $\mathcal C$; this would follow from $u + v + w \equiv 1$ which at present seems only to be derivable from analyticity considerations.  Duality implies e.g., $v_\e^B + w_\e^Y \equiv 1$ on $\mathcal C$, but we cannot go any further without color symmetry as in the site percolation on the triangular lattice case (\cite{stas_perc}) or some (asymptotic) color symmetry restoration as was established for the models in \cite{cardy}.
\end{remark}

Recalling that $C_0(\Omega, a, b, c, d)$ is equal to e.g., $u(d)$ with $d \in \mathcal A$, we now have 
\begin{thm}\label{cardy_sup}
For the models described in \cite{cardy} with the assumption $M(\partial \Omega) < 2$ (which includes the triangular site problem studied in \cite{stas_perc}, where the assumption on $\partial \Omega$ is unnecessary) Cardy's Formula can be established via an interior or sup--approximation, i.e., 
\[C_\e(\Omega_\e, a_\e, b_\e, c_\e, d_\e) \rightarrow C_0(\Omega, a, b, c, d)\]
if $(\Omega_\e)$ is an interior or sup--approximation to $\Omega$.
\end{thm}
\begin{proof}
For the site percolation model, this follows from \cite{stas_perc}, \cite{cardy}, \cite{beffara}, and Theorem \ref{bv_prop}.  For the model described in \cite{cardy}, the interior analyticity statement in sufficient generality is verified in \cite{pt1}, $\S$4.4.
\end{proof}

Finally, let us single out the cases that will be used in the proof of the Main Theorem in \cite{pt1}.  

\begin{cor}\label{sup_approx_conv}
Consider the models described in \cite{cardy} (which includes the triangular site problem studied in \cite{stas_perc}) on a bounded domain $\Omega$ with boundary Minkowski dimension less than two (if necessary) and two marked boundary points $a$ and $c$.  Suppose we have $\mathbb X_{[0, t]}^\e \rightarrow \mathbb X_{[0, t]}$ in the \textbf{\emph{Dist}} norm where $\mathbb X_{[0, t]}^\e$ is the trace of a discrete Exploration Process starting at $a$ and aiming towards $c$, stopped at some time $t$, then 
\[C_\e(\Omega_\e \setminus \mathbb X_{[0, t]}^\e, \mathbb X_t^\e, b_\e, c_\e, d_\e) \rightarrow C_0(\Omega \setminus \mathbb X_{[0, t]}, \mathbb X_t, b, c, d).\] 
\end{cor}

Further, it is possible to extract a slightly stronger statement which will be used in the proof of the Main Theorem in \cite{pt1}.  For the sake of \cite{pt1} we will state these results in the $\textbf{Dist}$ norm (c.f., Remark \ref{sup_dist}).  For purposes of clarity, we first state a lemma:

\begin{prop}\label{diag_conv}
Let us denote the type of (slit) domain under consideration by $\Omega^\gamma$ and abbreviate, by abuse of notation, e.g., $C_\e((\Omega^\gamma)_\e) := C_\e(\Omega_\e \setminus \gamma_\e([0, t]), \gamma_\e(t), b_\e, c_\e, d_\e)$ (but here, $\gamma$ could stand for other boundary pieces as detailed in Definition \ref{sup_defn}).  Then for any sequence $\gamma_n \rightarrow \gamma$ in the \textbf{\emph{Dist}} norm and any sequence $(\e_m)$ converging to zero, 
\[\lim_{n, m \rightarrow \infty} C_{\e_m}\left[(\Omega^{\gamma_n})_{\e_m}\right] = C_0(\Omega^\gamma),\]
regardless of how $n$ and $m$ tend to infinity.
\end{prop}
\begin{proof}
From Lemma \ref{sup_bound_value} we have that e.g., if $\gamma_{\e_m}^{(n)} \rightarrow \gamma_n$ is any sup--approximation, then $C_{\e_m}[(\Omega^{\gamma_n})_{\e_m}] \rightarrow C_0(\Omega^{\gamma_n})$.  The result follows by noting that $\gamma_{\e_m}^{(n)}$ is also a sup--approximation to $\gamma$ as \emph{both} $m, n \rightarrow \infty$.  We emphasize that the reason for such robustness of Lemma \ref{sup_bound_value} is because the proof is completely insensitive to how $\gamma_{\e}$ converges to $\gamma$ as $\e \rightarrow 0$.  All that is needed is that $\gamma_\e$ is sufficiently close to $\gamma$ and $\e$ is sufficiently small, which is inevitable if $\e$ is tending to zero and $\gamma_\e$ is tending to $\gamma$.
\end{proof}

\begin{cor}\label{point_equi_cont}
Considered the models described in \cite{cardy} (which includes the triangular site problem studied in \cite{stas_perc}) on a bounded domain $\Omega$ with boundary Minkowski dimension less than two (if necessary) and two marked boundary points $a$ and $c$.  Consider $\mathscr C_{a, c, \Delta}$, the set of L\"oewner curves which begin at $a$, are aiming towards $c$ but have not yet entered the $\Delta$ neighborhood of $c$ for some $\Delta > 0$.  Suppose we have $\gamma_\e \rightarrow \gamma$ e.g., in the \textbf{\emph{Dist}} norm,  then 
\[C_\e(\Omega_\e \setminus \gamma_\e([0, t]), \gamma_\e(t), b_\e, c_\e, d_\e) \rightarrow C_0(\Omega \setminus \gamma([0, t]), \gamma(t), b, c, d)\]
pointwise equicontinuously in the sense that
\begin{equation}\label{pt_equi_cont}\begin{array}{c} \forall \kappa > 0, ~~\forall \gamma \in \mathscr C_{a, c, \Omega}, ~~\exists \delta(\gamma) > 0, ~~\exists \mathcal \e_\gamma,\\ 
\\
\mbox{ such that }\\ 
\\
\forall \gamma^\prime \in \mathcal B_{\delta(\gamma)}(\gamma),~~ \forall \e \leq \mathcal \e_\gamma,\\ \\  
|C_\e((\Omega \setminus \gamma)_\e([0, t])), (\gamma(t))_\e, b_\e, c_\e, d_\e) - C_\e((\Omega \setminus \gamma^\prime)_\e([0, t])), (\gamma^\prime(t))_\e, b_\e, c_\e, d_\e)| < \kappa.\end{array}\end{equation}
Here $B_\delta(\gamma)$ denotes the \textbf{\emph{Dist}} neighborhood of $\gamma$.  
\end{cor}
\begin{proof}
This is immediate from Proposition \ref{diag_conv}.  Negation of the conclusion in the statement means that there exists a sequence $\gamma_n \rightarrow \gamma$ and $\e_n \rightarrow 0$ such that $|C_{\e_n}((\Omega^{\gamma_n})_{\e_n}) - C_\e((\Omega^{\gamma})_{\e_n})| > \kappa > 0$ for all $\e_n$, which clearly contradicts the fact that both of these objects converge to the limit $C_0(\Omega^\gamma)$.  

%This follows from a more quantitative reading of the proof of Lemma \ref{sup_bound_value} from which Theorem \ref{cardy_sup} follows: Given $\gamma$, Lemma \ref{top_consist} guarantees ``topological consistency'' of domains $\Omega \setminus \gamma$ and $\Omega \setminus \gamma^\prime$ for $\gamma^\prime$ at least some $\delta(\gamma)$ close to $\gamma$,
%
%now a small modification of the RSW type argument near the end of the proof of Lemma \ref{sup_bound_value} gives uniform in $\e$ estimates for crossing probabilities for $\e$ sufficiently small (see also Remark \ref{cont_to_boundary}).
%
%\mage{actually need to be a little careful: since cardy's formula correspond to \emph{non--trivial} boundary value, say $u$ on $\mathcal A$.  have $u_n(z)$ and $\tilde u_n(z)$; we want to conclude there exists $\delta$ and for $n$ sufficiently large, etc., that $|u_n(z) - \tilde u_n(z)|$ is small.}
%
%\mage{\emph{a priori} could have 
%non--local problems (e.g., blue crossing smashes into another part 
%of boundary of one approx.~domain but well--behaved for the 
%other), so not quite enough to ``seal off'' with RSW near $z$}
%
%\mage{however, since \emph{know} $u_n(z) \rightarrow C_0$ \emph{and} $\tilde u_n(z) \rightarrow C_0$, if no $\delta, \e$ as required exists, then have $z_n$, $\delta_n \rightarrow 0$ and $\e_n \rightarrow 0$, such that $|u_n(z_n) - \tilde u_n(z_n)| > \kappa$, some $\kappa$, contradicting the fact that $u_n(z_n) \rightarrow u(z) := \lim \tilde u_n(z_n)$.  thus have existential $\e$ which depends on rate at which $u_n$ and $\tilde u_n$ converges to $C_0$}
\end{proof}

\begin{remark}
We remark that \eqref{pt_equi_cont} holds even if ``$\e = 0$'' and thus implies continuity of Cardy's Formula in the ``\textbf{Dist} norm''.  However, we note that Lemma \ref{sup_bound_value}, being merely a limiting statement, would be highly inadequate if one had in mind some uniformity of the convergence or uniformity of the continuity.  
\end{remark}

\section*{\large{Acknowledgments}} 

\hspace{16 pt}The authors are grateful to the IPAM institute at UCLA for their hospitality and support during the ÒRandom Shapes ConferenceÓ (where this work began).  The conference was funded by the NSF under the grant DMS-0439872.  I.~B.~was partially supported by the NSERC under the DISCOVER grant 5810-2004-298433. L.~C.~was supported by the NSF under the grant DMS-0805486.  H.~K.~L was supported by the NSF under the grant DMS-0805486 and by the Dissertation Year Fellowship Program at UCLA.  
\vspace{.2 cm}

The authors would also like to thank Wendelin Werner for useful discussions which took place during the Oberwolfach conference \emph{Scaling Limits in Models of Statistical Mechanics} and which led to the present approach.  

%where $\gamma_\e$ converges to $\gamma$ in the \textbf{\emph{Dist}} norm.  
%
%pointwise equicontinuity of crossing probability with $\e$ in the game:
%
%Let $X$ be a second countable metric space and let $\mu_n$ be probability measures on $X$ weakly converging to $\mu$.  Suppose $f_n: X \rightarrow \mathbb R$ are uniformly bounded: $f_n \leq C$, and converge to $f$ pointwise equicontinuously, i.e.,
%\[ \forall \epsilon > 0, \forall x \in X, \exists \delta(x) > 0, \exists \mathcal N_x, \mbox{ such that } \forall y \in \mathcal B_{\delta(x)}(x), \forall n \geq \mathcal N_x, |f_n(x) - f_n(y)| < \epsilon, \]

\end{document}